\documentclass[12pt]{article}
\usepackage{amsmath} %Never write a paper without using amsmath for its many new commands
\usepackage{amssymb} %Some extra symbols
\usepackage{graphicx} %If you want to include postscript graphics
\usepackage[utf8]{inputenc}
\usepackage{amsthm}
\usepackage{color}
\usepackage{relsize}

\newtheorem{theorem}{Theorem}
\newtheorem{lemma}[theorem]{Lemma}
\newtheorem{cor}[theorem]{Corollary}
\newcommand{\cS}{\mathcal{S}}

\newcommand{\F}{\mathbb{F}}

\begin{document}

\title{Correlation, Linear Complexity, Maximum order Complexity on Families of binary Sequences}

\author{Zhixiong Chen$^1$, Ana I. G\'{o}mez$^2$,\\
 Domingo G\'{o}mez-P\'{e}rez\footnote{Corresponding author: domingo.gomez@unican.es.}~~$^3$, Andrew Tirkel$^{4}$\\
1. Key Laboratory of Applied Mathematics of Fujian Province University, \\ Putian University,  Putian, Fujian
351100, P. R. China\\
2. Universidad Rey Juan Carlos, Spain\\
3. Universidad de Cantabria, Spain\\
4. Scientific Technologies, Australia
}

\maketitle

\begin{abstract}
  Correlation measure of order $k$ is an important measure of randomness in binary sequences. This measure tries to look for dependence between several shifted version of a sequence.
  We study the relation between the correlation measure of order $k$ and another two pseudorandom measures: the $N$th linear complexity and the $N$th maximum order complexity. We simplify and improve several state-of-the-art lower bounds for
  these two measures using the Hamming bound as well as weaker bounds derived from it.
\end{abstract}

\section{Introduction}
\label{sec:orgc42ebce}

For a positive integer $N$, the \emph{$N$th linear complexity
  $L(\cS,N)$} of a binary sequence $\cS=(s_i)^{\infty}_{i=0}$ over the two-element finite field $\mathbb{F}_2=\{0,1\}$ is
the smallest positive integer
$L$ such that there are constants $c_0,c_1,\ldots,c_{L-1} \in \F_2$ with
\begin{equation}
  \label{eq:defLin}
  s_{i+L}=c_{L-1}s_{i+L-1}+\ldots+c_{0}s_{i},\quad \mbox{for }0 \le i < N - L.
\end{equation}
We use the convention $L(\cS,N)=0$ if $s_0=\ldots=s_{N-1}=0$ and
$L(\cS,N)=N$ if $s_0=\ldots=s_{N-2}=0\ne s_{N-1}$.
The $N$th linear complexity is a measure for the predictability of a
sequence and thus its unsuitability in cryptography.
If $\cS$ is $T$-periodic, we have
$ L(\cS,N)=L(\cS, 2T)$ for $N\ge 2T$. This number is the \emph{linear
complexity} of the sequence $\cS$.

Analogously, the $N$th \emph{maximum-order complexity} $M(\cS,N)$ of a binary sequence $\cS=(s_i)_{i=0}^\infty$ is defined as the smallest positive integer $M$
such that there is a polynomial $f(x_1,\ldots,x_M)\in \F_2[x_1,\ldots,x_M]$
with
$$
s_{i+M}=f(s_i,s_{i+1},\ldots,s_{i+M-1}),\quad \mbox{for }0 \le i < N - M,
$$
see~\cite{ja89,ja91,nixi14}. Again, if the sequence is $T$-periodic,
$M(\cS,N)=M(\cS,2T)$ for $N\ge 2T$. This is called
\emph{maximum-order complexity} of $\cS$.

Obviously, we have
$M(\cS,N)\le L(\cS,N)$,
so maximum-order complexity is a finer measure of pseudorandomness than linear complexity.

Let $k$ be a positive integer. The \emph{($N$th) correlation
  measure of order $k$} of $\cS$ is defined as
%\begin{equation}\label{legcorr}
$$C_{k}(\cS,N)=\max_{U,D}\left|\sum^{U-1}_{n=0}(-1)^{s_{n+d_{1}}+s_{n+d_{2}}+\ldots+s_{n+d_{k}}}\right|,$$
%\end{equation}
where the maximum is taken over all $U \leq N - k + 1$ and
$D = (d_1, d_2, \ldots , d_k)$ with
integers $0 \leq d_1 < d_2 < \ldots < d_k \le N - U$.
This is an adaptation to the binary case of the definition
concerning sequences over $\{-1,+1\}$, introduced by Mauduit
and S\'ark\"ozy~\cite{masa97}.

Brandst\"atter and Winterhof~\cite{branwin06} proved the following
relation between the $N$th linear complexity and the correlation
measures of order $k$:
\begin{equation}
  \label{eq:brandstatter}
L(\cS,N)\ge N - \max_{1\le k \le L(S,N)+1} C_k(\cS,N),
\quad\mbox{for }N\ge 1.
\end{equation}

Recently, I\c{s}{\i}k and Winterhof~\cite{IW2017} have derived an
analogous result concerning the $N$th maximum-order
complexity:
\begin{equation}
  \label{eq:Isik}
  M(\cS,N) \ge N- 2^{M(\cS,N)+1} \cdot \max_{1\le k \le M(\cS,N)+1} C_k(\cS,N),\quad\mbox{for } N\ge 1.
\end{equation}

Roughly speaking, any sequence with small correlation measure up to a
sufficiently large order $k$ must have a high $N$th  maximum-order
complexity (and hence $N$th linear complexity) as well.
For surveys on linear complexity and related measures of pseudorandomness, see~%
\cite{gy13,mewi13,ni03,sa07,towi07,wi10}.

The problem with these bounds is that they seem to be far from tight. Even if the correlation measure is close to
the expected value for a random binary sequence, the bounds above are far from expected.

Due to the constraints on the $N$th correlation measure given by Gyarmati and Mauduit~\cite{Gyarmati2012}, which implies that
the correlation measure is bigger than $\sqrt{N}$ for many orders,
the lower bound in Equation~\eqref{eq:brandstatter} is $2\sqrt{N}$. For $M(\cS,N)$, the lower bound can not be greater than $(\log N )/2$, see e.g. \cite{IW2017}.
Notice that the expected $N$th linear complexity of a random binary sequence is $N/2$~%
\cite{rueppel1985linear}. For the $N$th maximum-order complexity, the expected value is $2\log N$~\cite{ja89}.

In this work, we discuss the higher order correlation of
binary sequences, improving the lower bounds
%% on $N$th linear complexity and $N$th maximum-order
%% complexity
shown in Equations~\eqref{eq:brandstatter}
and~\eqref{eq:Isik}. Then, we review the literature and
improve the lower bounds on linear complexity and
maximum-order complexity of several known sequences.

Our results are based on the
Hamming bound on error-correcting codes
(see e.g.~\cite[Theorem
3.4.6]{Niederreiter2015}).
Additionally, we use the following definition for
the \emph{periodic correlation measure of order $k$} of a $T$-periodic binary sequence $\cS$,
$$
\theta_k(\cS) = \max_{D}\left|\sum^{T-1}_{n=0}(-1)^{s_{n+d_{1}}+s_{n+d_{2}}+\ldots+s_{n+d_{k}}}\right|,
$$
where $D=(d_1,\ldots, d_k)$, with
$0\le d_1 < d_2 < \ldots < d_k < T$.

A binary sequence $\cS$ is said to have a \emph{full peak}
in the aperiodic correlation measure of order $k$ if $C_k(\cS,N)= N-k+1$.
It has a \emph{half peak} if $C_{k}(\cS, N)\ge N/2$. The
same definitions apply also for $\theta_k(\cS)$, the
periodic correlation measure of order $k$.

We suppress ``of order $k$'' when referring to the correlation measure  when the order $k$ is clear from the context.

\section{Higher-order correlation measure}
\label{sec:org1c9e72b}

We prove below a link between the linear complexity of a
sequence and its correlation measure.
Before, we state a direct consequence of the Hamming
bound~\cite[Theorem 3.4.6]{Niederreiter2015}.

\begin{lemma}
\label{lemma:Hamming}
Let $p$ be a prime number and $C \subseteq \F_{p}^{T}$ a linear subspace
of dimension $d$ (i.e.\@ a linear code over $\mathbb{F}_p$).
If, for some integer $t > 0$,
$$
\sum_{i=0}^{\lfloor (t-1)/2\rfloor}{T\choose i }(p-1)^{i}> p^{T-d},
$$
there exists a nonzero vector $\vec{v}\in C$ with at most $t$ nonzero components.
\end{lemma}

The strong relation between cyclic codes and periodic
sequences allows using the previous lemma to relate the
linear complexity of a sequence with the existence of full peaks
in the periodic correlation measure.

\begin{theorem}
  \label{thm:periodic}
Let $\cS=(s_i)^{\infty}_{i=0}$ be a $T$-periodic binary sequence with
 linear complexity $L$.
If, for some integer $t > 0$,
$$
\sum_{i=0}^{\lfloor (t-1)/2\rfloor}{T\choose i }\ge 2^{L},
$$
the sequence has a full peak in the periodic correlation measure
$\theta_k(\cS)$ for some $k$ with $1 < k \le t$, i.e., $\theta_k(\cS)=T$.
\end{theorem}
\begin{proof}
Let $C \subseteq \F_{2}^{T}$ be the linear subspace generated by
$$
(s_{0},s_{1},\ldots, s_{T-1}),(s_{1},s_{2},\ldots, s_{0}),\ldots, (s_{T-1},s_{0}\ldots, s_{T-2}),
$$
i.e.\@ a sequence's period and all its shifted versions. We
denote by $C^{\bot}$ the \emph{orthogonal subspace} of $C$, i.e.
$$
C^{\bot} = \left\{ (c_{0},\ldots, c_{T-1}) \in \F_{2}^{T} ~~ : ~~ \sum_{i=0}^{T-1}c_{i}s_{n+i}
= 0, \ \forall n\ge 0 \right\}.
$$
%% It is trivial to check that the dimension of $C^{\bot}$ is $T - L$, and hence that of $C$ is $L$.
It is trivial to check that $\dim(C) = L$, and hence $\dim(C^{\bot}) = T - L$.
By Lemma~\ref{lemma:Hamming} (with $p=2$), there exists a vector in $C^\bot$ with exactly $k \le t$ nonzero components.
Let  $d_{1},\ldots, d_{k}$ be their indices,
so
$$
\sum_{j=1}^{k}s_{n+d_{j}} = 0,\quad \forall n\ge 0.
$$
This implies that there is a full peak in periodic correlation measure of order $k$.
\end{proof}
For the aperiodic correlation, we have the following result.

\begin{theorem} %{cor}
  \label{cor:aperiodic}
  Let $\cS=(s_i)^{\infty}_{i=0}$ be a $T$-periodic binary sequence with
 $N$th linear complexity $L(\cS, N)$.
If, for some integer $t > 0$,
$$
{\lfloor N/2\rfloor \choose t }\ge 2^{L(\cS,N)},
$$
the sequence has a half peak in the aperiodic correlation
measure $C_k(\cS,N)$ for some $k$ with $1 < k \le 2t$, i.e., $C_k(\cS,N)\geq N/2$.
\end{theorem}  %{cor}
\begin{proof}
  Suppose that the sequence satisfies  Equation~\eqref{eq:defLin}, which means that
  the first $L(\cS, N)$ elements and the recurrence generates the next $N - L(\cS, N)$.
  There are at most \(2^{L(\cS,N)}\) different sequences of length $N$ that can be
  generated by the same linear recursion.

One the other hand, any sequence $(y_{n})^{\infty}_{n=0}$ defined as
\begin{equation}
y_{n} = \sum_{j=1}^{t}s_{n+d_{j}},\quad\mbox{with }
0 \leq d_1 < \ldots < d_t < \lfloor N/2\rfloor ,
\end{equation}
can also be generated by that linear recursion. There are at least \({\lfloor N/2\rfloor\choose t}\) ways of choosing that shift set $\{d_j\}$.

Therefore, by hypothesis, there exist two different ordered list of shifts: \(\{d_{1},\ldots, d_{t}\}\) and \(\{e_{1},\ldots, e_{t}\}\) such that
\begin{equation}
\sum_{j=1}^{t}s_{n+d_{j}} = \sum_{j=1}^{k}s_{n+e_{j}} \implies
\sum_{j=1}^{t}(s_{n+d_{j}} - s_{n+e_{j}}) = 0,
\end{equation}
for $0\le n\le \lceil N/2\rceil \le  N-\max\{d_t, e_t\}$.
Then, there is a half peak in the $N$th correlation measure of order
at most \(2t\).
\end{proof}

The following result, which we state for its applications,
is a direct consequence of Theorem~\ref{cor:aperiodic}.
\begin{cor}
  \label{cor:result}
  Given any positive integers $K$ and $N$ with $K^{2}<N$.
  If a binary sequence $\cS$ satisfies
  $C_{k}(\cS, N) < N/2$ for every $k<K$, we have
  $$
  L(\cS, N)> \frac{1}{2}K(\log N + 1 - \log K) - \frac{1}{2}\log K + \delta,
  $$
  where $\delta$ is an absolute constant.
\end{cor}
\begin{proof}
  Because the result of Theorem~\ref{cor:aperiodic} does not hold,
  it must be the case that:
  \begin{equation}
  \label{eq:ineq1}
  2^{L(\cS,N)} \ge {\lfloor N/2\rfloor \choose K/2 },
  \end{equation}
  where substitute the combinatorial number by the Stirling approximation
  $$
  {\lfloor N/2\rfloor \choose K/2 } \approx
  \left( \frac{Ne}{K} \right)^{K/2}(2\pi K)^{-1/2}\varepsilon,
  $$
  where $\varepsilon$ is some positive constant.
  Taking logarithms at both sides of Equation~\eqref{eq:ineq1}, we get the result.
\end{proof}
We compare~Equation~\eqref{eq:brandstatter} and this new bound. First, whenever we can apply
the former, Corollary~\ref{cor:result} applies as well and
the lower bound is improved by a factor of $\log N$.
Also, it is enough to obtain a non-trivial bound for $C_{k}(\cS, N)$, a strong bound being no longer necessary.

These results have immediate application to the families of
binary sequences summarized in Table~\ref{tab:periodic}. For those sequences' definition, as well as parameters and properties,
see the book of Golomb and Gong~\cite{Golomb2005}.

\begin{table}[htbp]
\centering
\begin{tabular}{|c|c|c|c|}
\hline
Family & Period & Linear & Bound on $k$\\
 &  & complexity & for the existence\\
 &  &  & of a peak\\
\hline
$m$-sequences & $2^{\ell}-1$ & $\ell$ & 3\\
Small Kasami & $2^{\ell}-1$ & $3\ell/2$ & 5\\
Gold codes & $2^{\ell}-1$ & $2\ell$ & 7\\
Large Kasami & $2^{\ell}-1$ & $5\ell/2$ & 9\\
3-term trace & $2^{\ell}-1$ & $3\ell$ & 9\\
5-term trace & $2^{\ell}-1$ & $5\ell$ & 11\\
Welch-Gong & $2^{\ell}-1$ & $2^{\ell/3}+1$ & $(2^{\ell/3}+1)/\ell$\\
 &  &  & \\
\hline
\end{tabular}
  \caption{Different families of binary sequences, together with the upper bound on $k$ such that
there exists a peak in the periodic correlation measure of order $k$ according to Theorem~\ref{thm:periodic}}
\label{tab:periodic}
\end{table}

Results on Small Kasami and $m$-sequences have already been discovered by
Warner~\cite{warner-1994-tripl-correl,warner-1993-tripl-correl}. In the
case of the Gold codes, Adams~%
\cite{adams2004identification} presented some results regarding partial
peaks and conjectured on full peaks for order 9.
Boztas and Parampalli~\cite{Boztas2012} studied the third-order
correlation in order to assure the probability of intercept of Gold
codes.

We now enunciate a simple theorem of the same flavour for the $N$th maximum order complexity, improving the bound in Equation~\eqref{eq:Isik}.
\begin{theorem}   %{cor}
  \label{Thm:maximum}
  If a binary sequence $\cS$ satisfies
  $M(\cS,N)\le \log N - 2$, it has a half peak in the
  aperiodic correlation measure of order $2$, i.e.\@
  $C_2(\cS, N) \ge N/2$.
\end{theorem}     %{cor}
\begin{proof}
  In order to simplify the notation, $M=M(\cS,N)<\log N -2.$
  Under the hypothesis and since the first $N$ elements of
  the sequence can be generated by a polynomial with
  $M$ variables, i.e.
  $$
  s_{i+M}=f(s_i,s_{i+1},\ldots,s_{i+M-1}),\quad \mbox{for }0 \le i < N - M.
  $$
  By \cite[Propostion 2]{ja89}, the period of the sequence $(s_i)_{i=0}^{\infty}$ is
  less than $2^M$, see the explanation in the footnote\footnote{The idea is that the different possibilities for the
    tuples $(s_i,s_{i+1},\ldots,s_{i+M-1})$ is, at most, $2^M$. The tuple defines
    the next element, so this bounds the period of the sequence.}.

  This means that there exists $0\le d_1<d_2 < 2^M$ such
  that $s_{i+d_1} = s_{i+d_2}$ for $0\le i <N - M - d_2.$ The final step is
  $N - M - d_2> N - \log N +2 - N/4 > N/2$ and this finishes the proof.

\end{proof}

% /////////below to be continued//////////////

% \begin{cor}
% \label{cor:trivial}
% Let $\cS=(s_i)^{\infty}_{i=0}$ be a $T$-periodic binary sequence with
%  linear complexity profile $L(\cS, N)$ for $N\geq 0$.
%  If the inequality below holds for an integer $t>0$,
% $$
% {T\choose t}\ge p^{L(s_{n}, N)},
% $$
% then for some $1<k\le 2t$, $\theta_{k}((s_{n}, N - L(s_{n}, N))) = N -
% L(s_{n}, N)$.
% \end{cor}
% \begin{proof}
% Suppose that the sequence satisfies  \eqref{eq:defLin}, which means that
% generates  $N- L(\cS, N)$ elements, given the seed.
% There are at most $p^{L(s_{n},N)}$ different sequences of
% length $N- L(s_{n}, N)$ that can be generated by the same linear recursion.

% One the other hand, any sequence $(y_{n})$ defined as
% \begin{equation}
% \label{eq:2}
% y_{n} = \sum_{j=1}^{t}s_{n+d_{j}},\quad d_1, \ldots, d_t\le T.
% \end{equation}
% can be generated by the same linear recursion. There are ${T\choose t}$ ways of choosing different $d_{1},\ldots, d_{t}$ where $d_i<d_j$ if $i<j$.

% Therefore, there exists two different sets of cyclic shifts $d_{1},\ldots, d_{t}$ and $e_{1},\ldots, e_{t}$, we have that
% \begin{equation}
% \label{eq:2}
% \sum_{j=1}^{t}s_{n+d_{j}} = \sum_{j=1}^{k}s_{n+e_{j}} \implies
% \sum_{j=1}^{t}(s_{n+d_{j}} - s_{n+e_{j}}) = 0,
% \end{equation}
% which implies that there is a full peak in the correlation of order at
% most $2t$, because there could be some cancellations. This remark
% finishes the proof.
% \end{proof}

\section{Some applications}
\label{sec:applications}

\textbf{Hall's sextic residue sequence}.
The recent work of Aly and Winterhof~\cite{aly2019note}
studied Hall's sextic sequence, which is a binary
sequence with prime period $T = 1\mod 6$.  For such a period
and a primitive root modulo $T$, say $g$, {\em Hall's sextic
residue sequence} ${\cal H}=(h_n)^{\infty}_{n=0}$ is defined as
follows: let
\begin{equation}\label{cycl}
C_{\ell}=\{g^{6i+\ell}\ |\ 0 \leq i < (T - 1)/6\},\quad \ell=0,1,\ldots,5,
\end{equation}
be the cyclotomic cosets modulo $T$ of order $6$.
%Then it easy to see that $\Z^{*}_p = C_0 \cup C_1 \cdots  \cup C_5$.
Then, for $n \ge 0$,
\begin{equation}
\label{1}
h_n= \left\{\begin{array}{ll} 1,&  \mbox{if }  n\bmod T \in C_0 \cup C_1 \cup C_3;\\
                   0, &  \mbox{otherwise}.
    \end{array}\right.
\end{equation}

Hall's sextic sequence has several desirable features of pseudorandomness,
one of them being low correlation measure:

\begin{equation}
\label{Hall:bound}
C_{k}({\cal H}, N) = O\left( \left (\frac{14}{3}\right )^k k \sqrt{T}\log T \right).
\end{equation}
Using this bound and  the lower bound proved by
Brandst\"atter and Winterhof~\cite{branwin06}, it is shown in the
reference article~\cite{aly2019note} that the $N$th linear complexity is
$\Omega(\log T)$.
This is improved in the following result.

\begin{cor}
  For any $\varepsilon>0$, a sufficiently large $T$ and $N> 2 T^{1/2+\varepsilon}(\log T)^2$,
  the $N$th linear complexity of Hall's sextic sequence ${\cal H}$ satisfies
  $$
  L({\cal H}, N) \gg  (\log N)^2,
  $$
  where the implied constant depends on $\varepsilon$.
\end{cor}
\begin{proof}
  The correlation measure of order $k$ of Hall's sextic sequence is less than $N/2$ for $k \le \varepsilon \log T/8$,
  if $N\gg T^{1/2+\varepsilon}$. This is simple to see substituting in Equation~\eqref{Hall:bound},
  $$
  \left (\frac{14}{3}\right )^k k \sqrt{T}\log T \le
  \left (\frac{14}{3}\right )^{\log T/8} \log T \sqrt{T}\log T \le
  T^{1/2 + \varepsilon/2}(\log T)^2 < N/2.
  $$
  By Theorem~\ref{cor:aperiodic}, we have
  $$
  L({\cal H}, N)\ge \varepsilon \log T/16(\log N   - \log \varepsilon - \log \log N - 3)- \log \log N + \delta \gg
  (\log N)^2.
  $$
  This finishes the proof.
\end{proof}

\textbf{Fermat quotient threshold sequence}.
For prime $p$ and an integer $u$ with $\gcd(u,p)=1$, the {\it
Fermat quotient $q_p(u)$ modulo~$p$\/} is defined as the unique
integer with
$$
q_p(u) = \frac{u^{p-1} -1}{p} \pmod p, \qquad 0 \le q_p(u) < p.
$$
We also define
$$
q_p(kp) = 0,\qquad \mbox{for } k \in \mathbb{Z}.
$$
Note that $(q_p(u))$ is a $p^2$-periodic sequence modulo~$p$, so $T=p^2$. Then the \emph{binary threshold
sequence} ${\cal E}=(e_n)^{\infty}_{n=0}$ is defined by
$$
e_u=\left\{
\begin{array}{ll}
0, & \mathrm{if}\,\ 0\leq q_p(u)/p< \frac{1}{2};\\
1, & \mathrm{if}\,\ \frac{1}{2}\leq q_p(u)/p< 1.
\end{array}
\right.% \quad 1\le u \le p^2.
$$
Note that for which applications a discrepancy bound with arbitrary
shifts is needed. Most discrepancy
bounds on nonlinear pseudorandom numbers found in the literature
consider only equidistant  shifts.

Using the same techniques, Chen et al.\@~\cite{chen2010structure} proved a bound on the correlation measure. In Theorem~3 of that paper, they showed that
\begin{equation}
  \label{eq:bound_fermat}
  C_{2}({\cal E}, N) \ll p (\log p)^3.
\end{equation}
The following corollary gives a new lower bound on the $N$th linear complexity.
\begin{cor}
  For any $\varepsilon>0$, a sufficiently large $p$ and $N> 2p^{1+\varepsilon}(\log p)^3$,
the $N$th linear complexity of the binary threshold sequence ${\cal E}$ satisfies
$$
L({\cal E},N)\gg \log N,
$$
where the implied constant depends on $\varepsilon$.
\end{cor}
\begin{proof}
  Again, it is easy to see that if $N> 2p^{1+\varepsilon}(\log p)^3$, then
  the correlation of the sequence of order $2$ is less than $N/2$.
  By Corollary ~\ref{cor:result}, taking $K=3$ and using the bound in
  Equation~\eqref{eq:bound_fermat}, we get the result.
\end{proof}

This improves the bound of order $(\log N - \log p)/\log\log
p$, given by Chen et al.\@~\cite[Theorem
  4]{chen2010structure}. As shown by this result, even weak bounds lead to improvements on the correlation measure provides information about the linear complexity.\\

\textbf{Error linear complexity profile of sequences}.
Another application is to lower bound the $K$-error linear
complexity profile, i.e. the minimum linear complexity profile among
sequences differing from the studied one in at most $K$
entries.
In particular, let us bound the $K$-linear complexity of  ${\cal E}$ and ${\cal H}$.
\begin{cor}
  For $N<T (=p^2)$, the $N$th linear complexity of the binary threshold sequence ${\cal E}$, allowing at most $N/6$ entry switches, is greater than $\log N$.
\end{cor}
\begin{proof}
  Notice that a change in $N/6$ or fewer sequence elements increases the value of the correlation measure of order $2$ in $N/3$.
  This is trivial to see from the definition, because it modifies at most $N/3$ terms, so the correlation goes up by $N/3$.

  Together with the bound in
  Equation~\eqref{eq:bound_fermat}, we obtain the result.
\end{proof}
The proof of the next result follows the same path as the previous one.
\begin{cor}
  For $N<T$, the $N$th linear complexity of Hall's sextic sequence ${\cal H}$, changing at most $N/6$, is greater than $(\log N)^2$.
\end{cor}
In Table~\ref{tab:aperiodic-more}, we compare with previous
results the obtained bounds for the $N$th linear complexity of
several sequences.
The resulting bound by Theorem \ref{Thm:maximum}
on the $N$th maximum order complexity for all of
the sequences listed in the table is $\log N -2$.

\begin{table}[htbp]
\centering
\begin{tabular}{|l|l|l|}
\hline
Sequence                                                                                                                    & Previous lower bound & Corollary~\ref{cor:result}          \\ \hline
\begin{tabular}[c]{@{}l@{}}Logarithm\\ threshold \\ sequence\\ \cite{branwin06}\end{tabular} & $\log N/\log\log T$  & $\log N$             \\ \hline
\begin{tabular}[c]{@{}l@{}}Two-prime\\ generator\\ sequence \\ \cite{branwin06}\end{tabular}           & $N/\sqrt{T}$         & $\sqrt{N}\log N$     \\ \hline
\begin{tabular}[c]{@{}l@{}}Modified inverse\\ threshold \\sequence \\ \cite{chen2008modified}\end{tabular} & $\log N/\log\log T$  & $(\log N)\log\log T$ \\ \hline
\begin{tabular}[c]{@{}l@{}}Binary cyclotomic\\ sequence \\ \cite{chen2009construction}\end{tabular}                 & $\log N/\log\log T$  & $(\log N)\log\log T$ \\ \hline
\begin{tabular}[c]{@{}l@{}}Inversive \\ threshold sequence \\ \cite{niederreiter2007structure}\end{tabular}                 & $\log N/\log\log T$  & $(\log N)\log\log T$ \\ \hline
\end{tabular}
\caption{Bound comparison. The previous results are stated using  simplified notation, where $T$ stands for the period.}
\label{tab:aperiodic-more}
\end{table}

\section{Conclusions and Acknowledgments}
This paper presents generalizations of the results appearing in the articles
\cite{branwin06} and \cite{IW2017}. Thanks to these results, it is possible to use these results mount correlation attacks in systems using standard families  of binary sequences like Gold codes and Kasami families (see Table~\ref{tab:periodic}). The results regarding the aperiodic form of the correlation measure of order $k$ improve the lower bound on the $N$th linear complexity given several papers, as stated in Table~\ref{tab:aperiodic-more}. For those sequences, we provide new non-trivial lower bounds on the maximum order complexity.

Domingo G\'omez-P\'erez and Ana I.\@ G\'omez are supported by the Spanish \emph{Agencia Estatal de Investigaci\'on} project
\emph{Secuencias y curvas en criptograf\'{\i}a} (PID2019-110633GB-I00/AEI/10.13039/501100011033).

Z. Chen was partially supported by the National Natural Science
Foundation of China under grant No.~61772292, and by the Provincial Natural Science
Foundation of Fujian, China  under grant No.~2020J01905.

\end{document}